\documentclass{amsproc}
\usepackage{hyperref}
\newtheorem{thm}{Theorem}[section]
\newtheorem{cor}[thm]{Corollary}

\theoremstyle{definition}

\theoremstyle{remark}

\numberwithin{equation}{section}

\usepackage{tikz}
\usepackage{color}
\begin{document}
\title{A symmetric difference-differential Lax pair for Painlev\'e VI}

\author{Chris M. Ormerod}
\address{California Institute of Technology, Mathematics 253-37, Pasadena, CA 91125}
\email{cormerod@caltech.edu}
\author{Eric Rains}
\address{California Institute of Technology, Mathematics 253-37, Pasadena, CA 91125}
\email{rains@caltech.edu}


\begin{abstract}
We present a Lax pair for the sixth Painlev\'e equation arising as a continuous isomonodromic deformation of a system of linear difference equations with an additional symmetry structure. We call this a symmetric difference-differential Lax pair. We show how the discrete isomonodromic deformations of the associated linear problem gives us a discrete version of the fifth Painlev\'e equation. By considering degenerations we obtain symmetric difference-differential Lax pairs for the fifth Painlev\'e equation and the various degenerate versions of the third Painlev\'e equation.
\end{abstract}

\maketitle

\section{Introduction}

The Painlev\'e equations are second order nonlinear differential equations whose only movable singularities are poles \cite{PainleveProperty}. The most general case in the classification of the Painlev\'e equations is the sixth Painlev\'e equation, written in standard form as 
\begin{align}
\label{P6}\tag{$\mathrm{P}_{VI}$}\dfrac{\mathrm{d}^2y}{\mathrm{d}t^2} =& \frac{1}{2} \left(\frac{1}{y}+ \frac{1}{y-1}+ \frac{1}{y-t}\right) \left(\dfrac{\mathrm{d}y}{\mathrm{d}t} \right)^2 - \left(\frac{1}{t}+\frac{1}{t-1}+\frac{1}{y-t}\right) \dfrac{\mathrm{d}y}{\mathrm{d}t} \\
&+ \frac{y(y-1) (y-t)}{(t-1)^2 t^2} \left(\alpha +\frac{\beta  t}{y^2}+\frac{\gamma 
   (t-1)}{(y-1)^2}+\frac{\delta  (t-1) t}{(y-t)^2}\right),\nonumber
 \end{align}
where $\alpha, \beta, \gamma, \delta \in \mathbb{C}$ are parameters. For special values of the parameters the solutions of \eqref{P6} have been used to express classes of Einstein metrics \cite{Tod1994}, correlation functions in the 2D Ising model \cite{Jimbo1980}, express the eigenvalue distributions for certain ensembles of random matrices \cite{Forrester2006} and characterize the reductions of nonlinear wave equations \cite{SKdVP6I} and the self-dual Yang-Mills equations \cite{Mason1993}. As a mathematical object, \eqref{P6} possesses a group of B\"acklund transformations that is of affine Weyl type $D_4^{(1)}$ \cite{Okamoto:Studies:I}, possesses solutions expressible in terms of Gauss's hypergeometric function \cite{Iwasaki2013} and admits a surface of initial conditions with a rich geometric structure \cite{Inaba2003}. 

Identifying the sixth Painlev\'e equation in applications usually hinges upon making a correspondence with a known Lax pair. This is usually done by showing that a system arises as an isomonodromic deformation of an associated linear problem of the right type. The most celebrated example arises as the isomonodromic deformations of a second order linear differential equation with four Fuchsian singularities \cite{Fuchs2, Fuchs1}. This makes understanding the Lax pairs of \eqref{P6} critically important in applications. For this reason, the Lax pairs of \eqref{P6} have been the subject of many works \cite{Conte2007, Kapaev1999, YamadaNoumi:LaxP6}.

Finding a Lax pair for a given system can be a difficult or even impossible task. The starting point is usually an ansatz made about the properties of some associated linear problem, one then tries to show that a given system coincides with the system of isomonodromic deformations \cite{Jimbo:Monodromy1}. This ansatz can be guided by the geometry of the moduli space of linear problems, which may be moduli spaces of systems of linear differential equations (or connections) \cite{Boalch2001, Hitchin1987}, or moduli spaces of difference equations, which were the subject of an article by one of the authors \cite{Rains2013}. 

The goal of this paper is to present a new Lax pair for \eqref{P6} which has been derived as a {\em continuous} isomonodromic deformation of a system of linear {\em difference} equations. We may write this system in matrix form as
\begin{subequations}
\begin{align}
\label{linear1}Y(x+1) =& A(x) Y(x),
\end{align}
where $x$ is called a spectral variable. The isomonodromic deformations are with respect to a new independent variable $t$, in which the evolution in $t$ is governed by an auxiliary linear problem of the form
\begin{align}
\label{deform1}\dfrac{\mathrm{d} Y(x)}{\mathrm{d}t}=& \mathcal{B}(x)Y(x).
\end{align}
\end{subequations}
We call Lax pairs that take the form \eqref{linear1} and \eqref{deform1} difference-differential Lax pairs. A novel feature of the Lax pair for \eqref{P6} we present is the presence of an additional symmetry; that the solutions of \eqref{linear1} satisfy 
\begin{equation}\label{linsymY}
Y(x) = Y(-x- \sigma),
\end{equation}
where $\sigma \in \mathbb{C}$. We call Lax pairs of this form symmetric difference-differential Lax pairs. By considering degenerations of the Lax pair we present we also obtain symmetric difference-differential Lax pairs for
 \begin{align}
\label{P5}\tag{$\mathrm{P}_{V}$}\dfrac{\mathrm{d}^2y}{\mathrm{d}t^2} =& \left( \frac{1}{2y}+\frac{1}{y-1}\right) \left(\dfrac{\mathrm{d}y}{\mathrm{d}t} \right)^2 - \dfrac{1}{t} \dfrac{\mathrm{d}y}{\mathrm{d}t} + \frac{(y-1)^2 }{t^2} \left(\alpha y +\frac{\beta }{y}\right) \\
&+ \dfrac{\gamma y}{t} - \dfrac{1}{2} \dfrac{y(y+1)}{y-1},\nonumber
\end{align}
\begin{align}
\label{P31}\tag{$\mathrm{P}_{III}(D_6)$} \dfrac{\mathrm{d}^2y}{\mathrm{d}t^2} =& \dfrac{1}{y} \left( \dfrac{\mathrm{d}y}{\mathrm{d}t}\right)^2 - \dfrac{1}{t} \dfrac{\mathrm{d}y}{\mathrm{d}t} + \dfrac{\alpha y^2 + \beta}{t} + \gamma y^3 + \dfrac{\delta}{y},
\end{align}
\begin{align}
\label{P32}\tag{$\mathrm{P}_{III}(D_7)$} \dfrac{\mathrm{d}^2y}{\mathrm{d}t^2} =& \dfrac{1}{y} \left( \dfrac{\mathrm{d}y}{\mathrm{d}t}\right)^2 - \dfrac{1}{t} 
\dfrac{\mathrm{d}y}{\mathrm{d}t} - \dfrac{2y^2}{t^2}  + \dfrac{\beta}{4t} - \dfrac{1}{y},
\end{align}
\begin{align}
\label{P33}\tag{$\mathrm{P}_{III}(D_8)$} \dfrac{\mathrm{d}^2y}{\mathrm{d}t^2} =& \dfrac{1}{y} \left( \dfrac{\mathrm{d}y}{\mathrm{d}t}\right)^2 - \dfrac{1}{t} \dfrac{\mathrm{d}y}{\mathrm{d}t} + \dfrac{y^2}{t^2}-\dfrac{1}{t}.
\end{align}
From the classification of Painlev\'e equations by their surface of initial conditions \cite{Sakai:Rational} the three versions of the third Painlev\'e equation have surfaces of initial conditions with distinct symmetry groups, hence, should be considered as distinct cases. In particular, \eqref{P32} and \eqref{P33} do not admit the first Painlev\'e equation as a limit. A degeneration diagram is shown in Figure \ref{degen}.

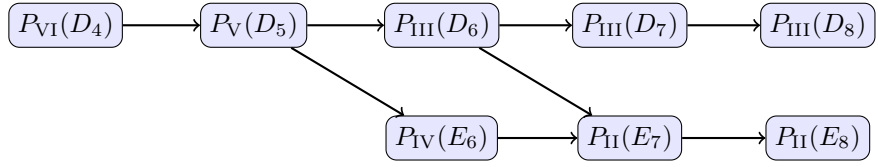
\begin{figure}[!ht]
\begin{tikzpicture}
\node[rectangle,fill=blue!10,rounded corners,draw=black] (P6) {$P_{\mathrm{VI}}(D_4)$};
\node[rectangle,fill=blue!10,rounded corners,draw=black,right of=P6, node distance=2.5cm] (P5) {$P_{\mathrm{V}}(D_5)$};
\node[rectangle,fill=blue!10,rounded corners,draw=black,right of=P5, node distance=2.5cm] (P31) {$P_{\mathrm{III}}(D_6)$};
\node[rectangle,fill=blue!10,rounded corners,draw=black,right of=P31, node distance=2.5cm] (P32) {$P_{\mathrm{III}}(D_7)$};
\node[rectangle,fill=blue!10,rounded corners,draw=black,right of=P32, node distance=2.5cm] (P33) {$P_{\mathrm{III}}(D_8)$};
\node[rectangle,fill=blue!10,rounded corners,draw=black,below of=P31, node distance=1.5cm] (P4) {$P_{\mathrm{IV}}(E_6)$};
\node[rectangle,fill=blue!10,rounded corners,draw=black,below of=P32, node distance=1.5cm] (P2) {$P_{\mathrm{II}}(E_7)$};
\node[rectangle,fill=blue!10,rounded corners,draw=black,below of=P33, node distance=1.5cm] (P1) {$P_{\mathrm{II}}(E_8)$};
\draw[->,thick] (P6) -- (P5);
\draw[->,thick] (P5) -- (P4);
\draw[->,thick] (P5) -- (P31);
\draw[->,thick] (P31) -- (P32);
\draw[->,thick] (P32) -- (P33);
\draw[->,thick] (P31) -- (P2);
\draw[->,thick] (P2) -- (P1);
\draw[->,thick] (P4) -- (P2);
\end{tikzpicture}
\caption{The degeneration diagram for the continuous Painlev\'e equations, including the degenerate cases of the third Painlev\'e equation. The brackets indicate the affine Weyl symmetry of the rational surface of initial conditions.\label{degen}} 
\end{figure}

Amongst the known Lax pairs for the Painlev\'e equations appearing in figure \ref{degen}, we note that there are differential-differential Lax pairs for each of the Painlev\'e equations \cite{Jimbo:Monodromy2, Ohyama2006}. However, there are relatively few works that consider continuous isomonodromic deformations of linear difference equations. Difference-differential Lax pairs for $\mathrm{P}_{I}$, $\mathrm{P}_{II}$, $\mathrm{P}_{IV}$, $\mathrm{P}_{V}$ and $\mathrm{P}_{VI}$ all appear in \cite{Adler1994} and the difference-differential Lax pairs for the versions of $P_{III}$ only recently appeared \cite{Kajiwara2015}.

We organize this article as follows: In \S 2 we present some background on known Lax pairs of the sixth Painlev\'e equation. In \S 3 we present our new Lax pair and a system of isomonodromic deformations. In \S 4 we present a general form of the Lax pair, which makes it clear how one degenerates the Lax pair to give \eqref{P5}, \eqref{P31}, \eqref{P32} and \eqref{P33}. In \S 5, we present a simple form of \eqref{P6} in which the limits to versions of \eqref{P5}, \eqref{P31}, \eqref{P32} and \eqref{P33} easily follow.

\section{Background and motivation}

In 1905, it was reported by R. Fuchs that if a scalar linear differential equation with four regular singular points, located at $0$, $1$, $t$ and $\infty$ and one apparent singularity at a value $y$, were to possess a monodromy representation that was independent of the singular point $t$, then $y$ necessarily satisfies \eqref{P6}. We may transform this to the second order Fuchsian differential equation, giving the following result. 

\begin{thm}[\cite{Fuchs2, Fuchs1}]
The isomonodromic deformations of the second order linear differential equation
\begin{align}\label{linear2}
\phi_{xx}(x,t) + \tau_1(x,t) \phi_x(x,t) + \tau_2(x,t) \phi(x,t) = 0,
\end{align}
where 
\begin{align*}
\tau_1(x,t) &= \dfrac{1-\kappa_0}{x} + \dfrac{1-\kappa_1}{x-1} + \dfrac{1-\kappa_t}{x-t} + \dfrac{1}{x-y},\\
\tau_2(x,t) &= \dfrac{1}{x(x-1)} \left( \dfrac{y(y-1)z}{x-y} - \dfrac{t(t-1)H_{VI}}{x-t} + \rho(\kappa_\infty+\rho) \right),\\
H_{VI} &= \dfrac{y(y-1)(y-t)}{t(t-1)} \left( z^2 - \left( \dfrac{\kappa_0}{y} + \dfrac{\kappa_1}{y-1} + \dfrac{\kappa_t-1}{y-t} \right) z + \dfrac{\rho(\kappa_\infty + \rho)}{y(y-1)} \right),\\
& \kappa_0 + \kappa_1 + \kappa_t + \kappa_\infty + 2 \rho = 1.
\end{align*}
is equivalent to $y(t)$ satisfying \eqref{P6} where
\begin{align*}
\alpha &= \dfrac{\kappa_\infty^2}{2}, \hspace{1cm} \beta = -\dfrac{\kappa_0^2}{2}, \hspace{1cm} \gamma = \dfrac{\kappa_1^2}{2}, \hspace{1cm} \delta = \dfrac{1-\kappa_t^2}{2}.
\end{align*}
\end{thm}

\begin{proof}
If we move $t$ in a continuous manner, we may use \eqref{linear2} to express the evolution in $t$ in the general form
\begin{align}\label{deform2}
\phi_t(x,t) + \sigma_1(x,t) \phi_x(x,t) + \sigma_2(x,t) \phi(x,t) = 0.
\end{align}
We can compute $\sigma_1(x,t)$ and $\sigma_2(x,t)$ by comparing the two ways of calculating $\phi_{xxt}(x,t)$; either by taking the derivative in $t$ of \eqref{linear2} or the second derivative in $x$ of \eqref{deform2}. In each case, these may be expressed as linear combinations of $\phi(x,t)$ and $\phi_x(x,t)$. Comparing these expressions is equivalent to Hamilton's equations, given by
\begin{align}\label{H6}
\dfrac{\mathrm{d} y}{\mathrm{d}t} =& \dfrac{\partial H_{VI}}{\partial z}, \hspace{2cm} \dfrac{\mathrm{d} z}{\mathrm{d}t} = -\dfrac{\partial H_{VI}}{\partial y}.
\end{align}
which is equivalent to $y(t)$ satisfying \eqref{P6} for the above parameters.
\end{proof}

The equations \eqref{linear2} and \eqref{deform2} in this calculation constitute a differential-differential Lax pair for \eqref{P6}. Alternatively, one may express the isomonodromic deformation problem in terms of matrices. It was the work of Jimbo, Miwa and Ueno that extended the work on isomonodromic deformations from systems with only Fuchsian singularities to systems of differential equations with higher order singularities. It is in this work we find the following parameterization of a Lax pair for \eqref{P6} in terms of matrices \cite{Jimbo:Monodromy2}.

\begin{thm}[\cite{Jimbo:Monodromy2}]
Consider the linear system of linear difference equations 
\begin{align}\label{linear3}
\dfrac{\mathrm{d}Y(x)}{\mathrm{d}x} = \left( \dfrac{A_0}{x} + \dfrac{A_1}{x-1} + \dfrac{A_t}{x-t} \right)Y(x), \\
A_0 + A_1 + A_t + \begin{pmatrix} \kappa_1 & 0 \\ 0 & \kappa_2 \end{pmatrix} = 0, \nonumber 
\end{align}
where
\begin{equation}
A_i = \begin{pmatrix} z_i + \theta_i & u_i z_i \\ \dfrac{z_i + \theta_i}{u_i} & -z_i \end{pmatrix},
\end{equation}
for $i = 0,1,t$. If $A(x) = (a_{i,j}(x))$ with
\begin{align}
a_{1,2}(x) = \dfrac{w(x-y)}{x(x-1)(x-t)}, \hspace{1cm} a_{1,1}(y) = z, 
\end{align}
then moving the parameter $t$ isomonodromically requires that $y = y(t)$ satisfies \eqref{P6} for 
\begin{align*}
\alpha &= \dfrac{(\kappa_1 - \kappa_2 -1)^2}{2}, \hspace{1cm} \beta = -\dfrac{\theta_0^2}{2}, \hspace{1cm} \gamma = \dfrac{\theta_1^2}{2}, \hspace{1cm} \delta = -\dfrac{1-\theta_t^2}{2}.
\end{align*}
\end{thm}

\begin{proof}
One method of deriving the system of isomonodromic deformations is to apply the Schlesinger equations, which, in this case specialize to
\begin{align}
\label{Schles1}\dfrac{\mathrm{d}A_0}{\mathrm{d}t} &= \dfrac{[A_t,A_0]}{t}, \\
\label{Schles2}\dfrac{\mathrm{d}A_1}{\mathrm{d}t} &= \dfrac{[A_t,A_1]}{t-1}, \\
\label{Schles3}\dfrac{\mathrm{d}A_t}{\mathrm{d}t} &= \dfrac{[A_0,A_t]}{t} + \dfrac{[A_1,A_t]}{t-1}.
\end{align}
This should be analogous to taking the deformation equation to be a linear system of the form \eqref{deform1}. The compatibility requires that the mixed derivatives agree, giving the condition
\begin{equation}\label{laxeqdiff2}
\dfrac{\partial A(x)}{\partial t} + A(x) \mathcal{B}(x) = \dfrac{\partial \mathcal{B}(x)}{\partial x} + \mathcal{B}(x)A(x).
\end{equation}
In fact, by taking 
\begin{equation}\label{deform3}
\mathcal{B}(x) =  -\dfrac{A_t}{x-t},
\end{equation}
then the residues of \eqref{laxeqdiff2} at $x=0$, $x=1$ and $x= t$ give \eqref{Schles1}, \eqref{Schles2} and \eqref{Schles3} respectively. Computing the entries of \eqref{laxeqdiff2} gives
\begin{align*}
\dfrac{\mathrm{d}y}{\mathrm{d}t} =& \dfrac{y(y-1)(y-t)}{t(t-1)} \left( 2z - \dfrac{\theta_0}{y} - \dfrac{\theta_1}{y-1} - \dfrac{\theta_t-1}{y-t} \right),\\
\dfrac{\mathrm{d}z}{\mathrm{d}t} =& \dfrac{1}{t(t-1)} \left( (2(1+t)y-t-3y^2)z^2\right. \\
 &\left.+ ((2y-1-t)\theta_0 + (2y-t)\theta_1 + (2y-1)(\theta_t-1))z - \kappa_1(\kappa_2+1) \right)
\end{align*}
Eliminating $z$ from this equation shows $y$ satisfies \eqref{P6} for the given parameters.
\end{proof}

The equations \eqref{linear3} and \eqref{deform1} where $\mathcal{B}(x)$ is specified by \eqref{deform3} constitute a differential-differential Lax pair for \eqref{P6}. We consider the scalar Lax pair, \eqref{linear2} and \eqref{deform2}, and the matrix Lax pair, \eqref{linear3} and \eqref{deform3} to be characteristically the same since it is trivial to convert \eqref{linear2} into a matrix form and \eqref{linear3} into a scalar form. 

To obtain a characteristically different Lax pair, let us consider the above calculations in a geometric framework. A useful geometric relaxation of the notion of a linear system of differential equations is the notion of a connection on a vector bundle, say $V$. We recover the notion of a matrix differential equation when we restrict our attention the case in which $V= \mathcal{O}_{\mathbb{P}_1}^n$ (copies of sheaves of holomorphic functions on $\mathbb{P}_1$). In this setting, \eqref{P6} may be interpreted as a flow on the 2-dimensional moduli space of second order Fuchsian linear differential equation with four singular points and specified exponents at those points. This moduli space has the structure of a rational surface with an anticanonical divisor $Y = -\mathcal{K}$ with a decomposition into five irreducible components, one of these components has multiplicity 2. This means the moduli space may be identified as the surface of initial conditions for \eqref{P6}.

Difference equations can also be viewed as maps between the fibres of a vector bundle \cite{ArinkinBorodin}, however, this identification still hinges upon identifying a map as matrix (through trivializing the vector bundle). Both the moduli space of matrices and the moduli space of maps between vector bundles may be interpreted as moduli spaces of sheaves on a smooth projective variety. These relaxations of the difference equations and their moduli spaces have been the subject of a recent study \cite{Rains2013}. These moduli spaces can have much the same structure as their continuous counterparts, hence, our motivation was to find a moduli space of difference equations that could also be identified with the surface of initial conditions for \eqref{P6}. The canonical flow on this surface should and does correspond to \eqref{P6}, as we will demonstrate.

One way to find a moduli space systems of linear difference equations of the same type as a given linear system of differential equations is via $Z$-transform. This is an invertible transformation that turns a linear differential equation into a linear difference equation. We suppose that a solution of \eqref{linear2} is given by a formal biinfinite power series,
\[
\phi(x,t) = \sum_{\chi \in \mathbb{Z}} a(\chi,t) x^\chi,
\]
then the coefficients, $a(\chi,t)$ satisfy a non-autonomous linear difference in the variable $\chi$. While applying this procedure to \eqref{linear2} directly results in a third order equation, we may transform this into a second order differential equation in which the entries are at most quadratic in the new spectral variable. This allows us to obtain a difference equation of the form \eqref{linear1} where
\[
A(x) = A_0 + A_1 x + A_2x^2,
\]
with the properties that $A_2$ has eigenvalues $1$ and $t$, the trace of $A_1$ is constant and the determinant of $A(x)$ is of quartic.  We let $a_1, \ldots, a_4$ be the roots of $\det A(x)$, in which case we write
\begin{align}\label{detassym}
\det A(x) &= t(x-a_1)(x-a_2)(x-a_3)(x-a_4),\\
&= t ( x^4 - \zeta_1 x^3 + \zeta_2x^2 - \zeta_3 x+ \zeta_4), \nonumber
\end{align}
where $\zeta_1, \ldots, \zeta_4$ are the elementary symmetric functions. The moduli space of such matrices up to gauge equivalence is a rational surface. 

Due to the invertible nature of the $Z$-transform, and the results of \cite{Inaba2003}, we expect this moduli space to coincide with the space of initial conditions for \eqref{P6} on an open subset of the surface. Exploiting some freedom in how we choose $A(x)$ allows us to simplify the associated linear problem to the following 
\begin{align}\label{AdP5}
A(x) = \begin{pmatrix} (x-z)(x-\tau_1) + y_1 & w(x-z) \\ \dfrac{\tau_3 x + \tau_4}{w} & t(x-z)(x-\tau_2) + y_2 \end{pmatrix},
\end{align}
where we may solve the relations a the coefficients of $x$ in \eqref{detassym} by letting
\begin{subequations}\label{tauvals}
\begin{align}
\tau_1 &= - \sigma - z, \\
\tau_2 &= \zeta_1 + \sigma - z,\\ 
\tau_3 &= \dfrac{t \zeta_4- y_1 y_2}{z^2} - \dfrac{t \zeta_3}{z} + t \tau_1\tau_2 + t y_1+y_2 + tz(\tau_1 + \tau_2),\\
\tau_4 &= \frac{t \zeta_4 -\left(y_1+\tau _1 z\right) \left(t \tau _2 z+y_2\right)}{z},
\end{align}
\end{subequations}
with 
\begin{equation}\label{y12vals}
y_1 = (z-a_1)(z-a_2)y, \hspace{1cm} y_2 = \dfrac{t(z-a_3)(z-a_4)}{y}.
\end{equation}
This forms a suitable parameterization of the given linear system of difference equations in terms of three variables, $y$, $z$ and $w$. The pair $(y,w)$ are sometimes called the spectral coordinates \cite{Dzhamay2007} and $w$ is a gauge factor. It is worth noting that the variable $\sigma$ plays a role in the asymptotics of the solutions of \eqref{linear1}, if we let
\begin{align}
d_1 = \sigma, \hspace{1cm} d_2 = - \sigma - \zeta_1,
\end{align}
then $d_1$ and $d_2$ are two of the characteristic constants of \eqref{linear1} in \cite{Borodin:connection}.

\begin{thm}
The continuous isomonodromic deformation of the linear system specified by \eqref{linear1} with $A(x)$ given by \eqref{AdP5}, \eqref{tauvals} and \eqref{y12vals} in $t$ requires that $y(t)$ satisfies \eqref{P6} for the parameters
\begin{align*}
\alpha &= \dfrac{(a_1-a_2)^2}{2}, \hspace{1cm} \beta = -\dfrac{(a_3-a_4)^2}{2}, \\ 
\gamma &= \dfrac{(a_1+a_2+\sigma)^2}{2}, \hspace{1cm} \delta = -\dfrac{(a_3+a_4+\sigma)(2+a_3+a_4+\sigma)}{2}.
\end{align*}
\end{thm}

\begin{proof}
While there are very few works on the continuous isomonodromic deformations of linear systems of {\em difference} equation, we follow the work of \cite{Adler1994} by parameterizing the isomonodromic deformations via \eqref{deform2} where $B(x)$ is linear. The compatibility between equations of the form \eqref{linear1} and \eqref{deform2} may be written 
\begin{equation}\label{Laxasym}
\dfrac{\mathrm{d}A(x)}{\mathrm{d}t} + A(x)\mathcal{B}(x) - \mathcal{B}(x+1)A(x) = 0.
\end{equation}
We find that \eqref{Laxasym} is overdetermined when $B(x)$ is linear, hence, we may write
\begin{align*}
\label{Bform2} \mathcal{B}(x) = \dfrac{x}{t} \begin{pmatrix} 0 & 0 \\ 0 & 1 \end{pmatrix} + \dfrac{1}{t(t-1)} \begin{pmatrix} 0 & w \\ \dfrac{\tau_3}{w} & 0 \end{pmatrix},
\end{align*}
where the other relations give us
\begin{subequations}\label{H6eq}
\begin{align}
z' =& \dfrac{y_1-y_2}{t(t-1)},\\
y' =&  \dfrac{y^2(a_1+a_2-2z)+t(a_3+a_4-2z) + y(\sigma-t(\zeta_1+\sigma-2z)+2z)}{t(t-1)},\\
\dfrac{w'}{w} =& -\dfrac{\sigma+t(\zeta_1+\sigma-z)+z}{t(t-1)}.
\end{align}
\end{subequations}
Eliminating $z$ from this system gives us the result. 
\end{proof}

The equations \eqref{linear1} and \eqref{deform1} for $A(x)$ and $\mathcal{B}(x)$ given by \eqref{AdP5} and \eqref{Bform2} respectively constitutes a difference-differential Lax pair for \eqref{P6}. This result is essentially a different parameterization of the Lax pair that appeared in \cite{Adler1994}. 

The linear system of the form \eqref{linear1} with $A(x)$ given by \eqref{AdP5} is also connected to a known Lax pair for the discrete version of the fifth Painlev\'e equation \cite{Dzhamay2007}. We consider the discrete version of the fifth Painlev\'e equation to be the map 
\[
\left( \begin{array}{c c}
a_1 & a_2 \\
a_3 & a_4 \end{array}\, \sigma; y,z \right)  \to \left( \begin{array}{c c}
a_1 +1& a_2 +1 \\
a_3 & a_4 \end{array}\, \sigma-1; \tilde{y},\tilde{z} \right) 
\]
where the values of $(\tilde{y},\tilde{z})$ are related to $(y,z)$ via
\begin{align}
\begin{array}{c}
\tilde{z} + z = a_1+a_2+1- \dfrac{a_3t +a_4 t+\sigma t}{\tilde{y}-t}+\dfrac{a_1+a_2+\sigma+1}{\tilde{y}-1},\\
y \tilde{y} = \dfrac{t (z-a_3)(z-a_4)}{(z-a_1)(z-a_2)}.
\end{array}\tag{$d$-$\mathrm{P}_{V}$}\label{dP5}
\end{align}
This transformation is a B\"acklund transformation of \eqref{P6} which commutes with the time evolution of \eqref{P6}.

The way in which a Lax pair for \eqref{dP5} arises is that $a_1,\ldots, a_4, \sigma$ appear as variables that change by integer shifts under the action of discrete isomonodromic deformations \cite{Borodin:connection}. A discrete isomonodromic deformation is a transformation induced by a gauge transformation of the form
\begin{equation}\label{discisomonodromy}
\tilde{Y}(x) = R(x)Y(x),
\end{equation}
where $R(x)$ is some rational matrix. We expect $\tilde{Y}(x)$ to satisfy an equation of the form
\[
\tilde{Y}(x+1) = \tilde{A}(x) \tilde{Y}(x),
\]
where $\tilde{A}(x)$ and $A(x)$ are related by
\begin{equation}\label{ddcomp}
\tilde{A}(x)R(x) - R(x+1)A(x) = 0.
\end{equation}
Our nonautonomous difference equation arises from identifying this as a map from the moduli spaces containing $A(x)$ to the  moduli space of matrices containing $\tilde{A}(x)$ in some canonical coordinate system. 

\begin{thm}
Given \eqref{linear1} with $A(x)$ given by \eqref{AdP5}, the discrete isomonodromic deformation induced by \eqref{discisomonodromy} when $R(x)$ takes the form
\[
R(x) = \dfrac{xI + R_0}{(x-a_1-1)(x-a_2-1)}.
\]
is equivalent to \eqref{dP5}.
\end{thm}

\begin{proof}
We may derive the entries of $R_0$ in a number of ways, we may compute the residues of \eqref{ddcomp} at the points $x= a_1$, $x=a_2$, $x= a_1+1$ and $x= a_2+1$, or compare the expansions in $x$ around $x=0$ and $x=\infty$. We find a suitable form of $R_0$ to be
\[
R_0 = \begin{pmatrix}  z -1-a_1-a_2+ \dfrac{t(a_3+a_4+\sigma)}{\tilde{y}-t} & \dfrac{\tilde{w}-w}{t-1}\\
 \dfrac{\tau_3\tilde{w}-w\tilde{\tau}_3}{(t-1)w \tilde{w}} & \dfrac{t(a_3+a_4-1+\sigma-z)+(1+z)\tilde{y}}{t-\tilde{y}}\end{pmatrix}. 
\]
The overdetermined compatibility of \eqref{ddcomp} results in \eqref{dP5} and
\[
\dfrac{\tilde{w}}{w} = \dfrac{\tilde{y}-1}{\tilde{y}-t}.
\]
which is the equation satisfied by the variable encapsulating the gauge freedom.
\end{proof}

We consider the equations \eqref{linear1} with $A(x)$ given by \eqref{AdP5} and \eqref{discisomonodromy} to be a difference-difference Lax pair. This Lax pair was a result of Arinkin and Borodin \cite{ArinkinBorodin} (see also \cite{Dzhamay2007}). 

\section{Symmetric difference-differential Lax pair}

One of the corollaries of the work on moduli spaces of difference equations was the existence of two moduli spaces of systems of linear difference equations of second order with the same structure as the space of initial conditions for \eqref{P6}. One of these moduli spaces corresponded to a difference-differential Lax pair without symmetry, as given by Adler \cite{Adler1994}, while the other possesses some symmetry with 4 pairs of symmetric singular points. In this section, we provide a parameterization of the resulting Lax pair that gives \eqref{P6} as a system of continuous isomonodromic deformations.

We recently outlined a discrete version of the Garnier system in which the solutions of the associated linear problem are even functions \cite{Ormerod2016}. A small generalization of this is to consider system of difference equations symmetric around some value. We wish to consider these difference equations as pairs of equations of the form
\begin{subequations}\label{linearasym}
\begin{align}
\label{linsyma}Y(-x-\sigma-1) &= A(x) Y(x),\\
\label{linsymb}Y(-x-\sigma) &= Y(x),
\end{align}
\end{subequations}
for some $\sigma \in \mathbb{C}$. The combination of \eqref{linsyma} and \eqref{linsymb} recovers \eqref{linear1}, hence, we consider this to be a system of linear difference equations. This extra constraint is motivated by the additional structure that many elliptic hypergeometric functions and elliptic biorthogonal functions possess.  

For the system defined by \eqref{linearasym}, by \eqref{linsymb}, it is easy to show that  $A(x)$ must necessarily satisfy the condition
\begin{equation}\label{symA}
A(x) A(-x-\sigma-1) = I.
\end{equation}
As discussed in \cite{Ormerod2016}, we may always write $A(x)$ in the form
\begin{align}\label{ABfactor}
A(x)  = B(-x-\sigma-1)^{-1} B(x),
\end{align}
where $B(x)$ is rational. Without loss of generality, we may multiply by scalar factors so that $B(x)$ is a polynomial. Furthermore, gauge transformations of the form \eqref{discisomonodromy} have the effect of multiplying $B(x)$ only by a factor on the right, whereas $B(x)$ itself is only defined up to multiplication by a factor on the left. In particular, matrix $B(x)$ from \eqref{ABfactor} is subject to constant gauge transformations on the left and the right. 

To make a correspondence the previous section we let
\begin{equation}\label{detB}
\det B(x) = t\left(x - a_1\right)\left(x - a_2\right)\left(x - a_3\right)\left(x - a_4\right).
\end{equation}
We may use the gauge freedom on the left and the right to simplify all but one entry, which we choose to be the $(2,1)$ entry. This means that we may generally take $B(x)$ to be of the form
\begin{equation}\label{Bform4}
B(x) = \begin{pmatrix} 
(x-\lambda_1)(x-z) + y_1 &w(x-z)\\ 
\dfrac{b_0 + b_1x + b_2 x^2 + b_3 x^3}{w} & (x-\lambda_2)(x-z) + y_2
\end{pmatrix}, 
\end{equation}
where the $b_i$ and $y_j$ for $i = 1,2,3,4$ and $j = 1,2$ are specified up to one parameter by \eqref{detB}. These values are
\begin{align*}
b_ 0 =& \dfrac{t \zeta_4 - (y_1+z\lambda_1)(y_2+z\lambda_2)}{z},\\
b_1 =& t \zeta_1 z- t \zeta_2 -t z^2 +\lambda _1 \lambda _2+y_1+y_2+\left(\lambda _1+\lambda _2\right) z,\nonumber \\
b_2 =& t(\zeta_1-z)-\lambda_1 -\lambda_2 -z, \nonumber \\
b_3 =& 1-t,
\end{align*}
where we have used the same elementary symmetric functions, $\zeta_1,\ldots, \zeta_4$, from the previous section. By suitable gauge transformations, we may take $\lambda_1 = \lambda_2 = 0$, since they are constant with respect to the continuous isomonodromic deformations. We do not simplify in this way because they satisfy some difference equation with respect to the discrete isomonodromic deformations. They essentially play the same role as $w$ in the previous section, i.e., they are variables encapsulating some left and right gauge freedom.

To obtain \eqref{P6} we let
\[
y_1 = (z-a_1)(z-a_2)y, \hspace{1cm} y_2 = \dfrac{t(z-a_3)(z-a_4)}{y}.
\]
These values of $b_i$ and $y_j$ for $i = 0,\ldots, 4$ and $j = 1,2$ are sufficient to ensure that \eqref{detB} is satisfied. 

\begin{thm}
The continuous isomonodromic deformations of \eqref{linear1} where $A(x)$ is given by \eqref{ABfactor} and \eqref{Bform4} is equivalent to $y(t)$ satisfying \eqref{P6} for
\begin{align*}
\alpha &= \dfrac{(a_1-a_2)^2}{2}, \hspace{1cm} \beta = -\dfrac{(a_3-a_4)^2}{2}, \\ 
\gamma &= \dfrac{(a_1+a_2+\sigma)^2}{2}, \hspace{1cm} \delta = -\dfrac{(a_3+a_4+\sigma)(2+a_3+a_4+\sigma)}{2},
\end{align*}
\end{thm}

\begin{proof}
The goal is to relate the isomonodromic deformations to an evolution equation for $B(x)$, which, if done correctly, preserves the symmetry condition for $A(x)$. In terms of $Y(x)$, the isomonodromic deformations are specified by an equation of the form
\begin{align}\label{symdeform2}
\dfrac{\mathrm{d}Y(x)}{\mathrm{d}t} = B_r(x)Y(x),
\end{align}
where $B_r(x)$ must satisfy $B_r(x) = B_r(-\sigma-x)$ due to the symmetry condition. We also note that $B(x)$ is only determined up to scalar matrix multiplication on the left, by some factor $B_l(x)$ satisfying $B_l(x) = B_l(-x-\sigma-1)$. These two conditions result in a consistency condition in which we have matrices acting upon the left and right of $B(x)$. The resulting symmetric analogue of \eqref{Laxasym} is 
\begin{equation}\label{comp}
\dfrac{\mathrm{d}B(x)}{\mathrm{d}t} = B_{l}(x) B(x) - B(x) B_{r}(x).
\end{equation}
When $B_l(x)$ and $B_r(x)$ are linear in $x(x+\sigma)$ and $x(x+\sigma+1)$ we have an overdetermined system with $B_l(x)$ and $B_r(x)$ given by 
\begin{align}
B_l(x) =& \begin{pmatrix} 0 & 0 \\ \dfrac{(1+\sigma)\lambda_2}{tw}-\dfrac{x(x+\sigma+1)}{tw}  & \dfrac{\sigma+1}{t} - \dfrac{w'}{w} \end{pmatrix} \\
&+ \dfrac{1}{1-t}\begin{pmatrix} a_1+a_2-z-\dfrac{\lambda_2}{t} & -\dfrac{w}{t} \\ -\dfrac{b_1+(b_2+b_3z)(z+\lambda_2)}{tw} - \dfrac{b_3z'}{w} & z-a_3-a_4+ \dfrac{\lambda_2}{t} \end{pmatrix}, \nonumber \\
B_r(x) =& \begin{pmatrix} \dfrac{\sigma}{t} & 0 \\ \dfrac{\sigma \lambda_1}{tw}- \dfrac{x(x+\sigma)}{tw} & -\dfrac{w'}{w} \end{pmatrix} \\
&- \dfrac{1}{1-t}\begin{pmatrix} a_3+a_4 - z - \dfrac{\lambda_1}{t}& \dfrac{w}{t} \\  \dfrac{b_1+(b_2+b_3z)(z+\lambda_1)}{tw} - \dfrac{b_3z'}{w}& z-a_1-a_2+ \dfrac{\lambda_1}{t}\end{pmatrix},\nonumber 
\end{align}
and the remaining conditions in \eqref{comp} give us
\begin{align}
z' =& \dfrac{y_1-y_2}{t(t-1)},\\
y' =&  \dfrac{1}{t(t-1)}\left( y^2(a_1+a_2-2z)+t(a_3+a_4-2z)\right. \\
& \left.+ y(\sigma-t(a_1+a_2+a_3+a_4+\sigma-2z)+2z) \right).\nonumber
\end{align}
These two equations coincides precisely with \eqref{H6eq} from the previous section. Following the same steps, we eliminate $z$ to show that $y$ satisfies \eqref{P6} for the given choice of parameters.
\end{proof}

The above shows that \eqref{linsyma}, \eqref{linsymb} and \eqref{symdeform2} constitutes a Lax pair for \eqref{P6}. However, due to the way in which we factor $A(x)$, the consistency condition, \eqref{comp}, is more naturally expressed in terms of $B(x)$ rather than $A(x)$.

One of the most interesting aspects of the symmetric Lax pairs is the existence of classes of symmetries that do not appear naturally in the work of Schlesinger transformations \cite{Jimbo:Monodromy2}. It is very natural to ask what the analogue of the Schlesinger transformations look like in the symmetric setting, in particular, how does \eqref{dP5} arise as a discrete isomonodromic deformation of a symmetric system of difference equations.

We find that the evolution of \eqref{dP5} is expressible in terms of the composition of three transformations, none of which simply shift the variables in the same way as the Schlesinger transformations in the asymmetric setting. It should be noted that while the composition of the transformations commutes with the time evolution of \eqref{P6}. It would be interesting to relate these transformations to a set of known generators of the associated affine Weyl group $D_4^{(1)}$. 

One interesting consequence of relating the parameters of \eqref{dP5} to \eqref{linsyma} and \eqref{linsymb} is that the value of  $\sigma$ is shifted with evolution of \eqref{dP5}, signifying that we seek a transformation in which the resulting symmetry is changed. Suppose we consider a transformation of the form \eqref{discisomonodromy}, where $R(x) = B(-x-\sigma)$. If $A(x)$ is given by \eqref{ABfactor}, then by using \eqref{ddcomp} we find that $\tilde{A}(x)$ is given by
\begin{align*}
\tilde{A}(x) =  B(x) B(-x-\sigma)^{-1}.
\end{align*}
This transformation changes the symmetry of the system as we notice that $\tilde{A}(x)$ satisfies 
\[
\tilde{A}(x) \tilde{A}(-x-\tilde{\sigma}-1) = I.
\]
If we denote the new value of $\sigma$ by $\tilde{\sigma}$, then the above indicates that $\tilde{\sigma} = \sigma -1$. By suitably multiplying by the determinant, we have that the relevant transformation is given by
\begin{align*}
\tilde{B}(x) =& t(x+\sigma+ a_1)(x+\sigma+ a_2)\\ 
&(x+\sigma+ a_3)(x+\sigma+ a_4)B(-x-\sigma)^{-1},
\end{align*}
which also corresponds to the adjoint. This means that the effect on the parameters, $y$ and $z$, are given by
\begin{align*}
\tag{$S_1$}\label{S1}&\left( \begin{array}{c c}
a_1 & a_2 \\
a_3 & a_4 \end{array}\, \sigma; y,z \right)\\  &\to \left( \begin{array}{c c}
-a_1-\sigma & -a_2-\sigma \\
-a_3-\sigma & -a_4-\sigma \end{array}\, \sigma-1;\dfrac{t}{y} \dfrac{(z-a_3)(z-a_4)}{(z-a_1)(z-a_2)} ,-\sigma-z \right).
\end{align*}
The auxiliary variables, $\lambda_1$ and $\lambda_2$ satisfy $\tilde{\lambda}_1= -\sigma-\lambda_2$ and $\tilde{\lambda}_2=-\sigma-\lambda_1$ respectively.

For the second transformation, we consider a matrix acting on the left of $B(x)$. As we explored in a recent paper on discrete Garnier systems, we seek a transformation in which we take two parameters, in the case that is relevant, we take $a_3$ and $a_4$, and send them to $-a_3-\sigma-1$ and $-a_4-\sigma-1$ respectively. Note that this matrix is easily calculated from the Lax equation 
\begin{align}\label{compl}
(x-a_3)(x-a_4)\tilde{B}(x) = R_l(x) B(x),
\end{align}
which is induced by a matrix $R_l(x)$ with the property that $R_l(x)= R_l(-x-\sigma-1)$. This calculation shows
\begin{align*}
&R_l(x) = \begin{pmatrix} (x-a_4)(x+\sigma+1+a_4) & 0 \\ 0 & (x-a_3)(x+\sigma+1+a_3)\end{pmatrix} \\
&+ (z-\tilde{z})\begin{pmatrix}  z - a_3 - \dfrac{y(a_4-\lambda_2)}{t} & \dfrac{wy}{t} \\ & \\ 
-\dfrac{t}{wt} \left(z-a_3-\frac{y(a_4-\lambda_2)}{t} \right)\left(z-a_4-\frac{y(a_3-\lambda_2)}{t} \right) & a_3-z+\dfrac{y(a_4-\lambda_2)}{t}\end{pmatrix}
\end{align*}
which then induces the transformation
\begin{align}
\tag{$S_2$}\label{S2}&\left( \begin{array}{c c}
a_1 & a_2 \\
a_3 & a_4 \end{array}\, \sigma; y,z \right)\\  &\to \left( \begin{array}{c c}
a_1 & a_2 \\ -a_3-\sigma-1 & -a_4-\sigma-1 \end{array}\, \sigma; y , z + \dfrac{t(1+\sigma+a_3+a_4)}{y-t} \right)\nonumber.
\end{align}
where the variables $\lambda_1$ and $\lambda_2$ become $\lambda_1 + y(z-\tilde{z})$ and $\lambda_2 + t^{-1} y(z-\tilde{z})$ respectively where by $\tilde{z}$ we mean the transformed value of $z$. 

The last transformation to form \eqref{dP5} is a transformation that acts on the right of $B(x)$. We take $a_1$ and $a_2$ and send these values to $-a_1-\sigma$ and $-a_2 -\sigma$ respectively. The relevant Lax equation for this transformation is given by 
\begin{equation}\label{Laxr}
(x-a_1)(x-a_2)\tilde{B}(x) = B(x)R_r(x),
\end{equation}
where $R_r(x) = R_r(-x-\sigma)$. It easy to calculate this matrix from the properties of $B(x)$ and \eqref{Laxr}. This calculation leads us to 
\begin{align*}
&R_r(x) = \begin{pmatrix} (x-a_1)(x+\sigma+a_1) & 0 \\ 0 & (x-a_2)(x+\sigma+a_2) \end{pmatrix}\\
& + (z-\tilde{z}) \begin{pmatrix} a_1-z+\dfrac{a_2-\lambda_1}{y} & \dfrac{w}{y} \\ -\dfrac{y}{w}\left(a_1-z+\dfrac{a_2-\lambda_1}{y} \right)\left(a_2-z+\dfrac{a_1-\lambda_1}{y} \right) & z-a_1- \dfrac{a_2-\lambda_1}{y} \end{pmatrix},
\end{align*}
with the transformation being 
\begin{align}
\tag{$S_3$}\label{S3}&\left( \begin{array}{c c}
a_1 & a_2 \\
a_3 & a_4 \end{array}\, \sigma; y,z \right)\to \left( \begin{array}{c c}
-a_1-\sigma & -a_2-\sigma \\ a_3 & a_4 \end{array}\, \sigma; y , z - \dfrac{(\sigma+a_1+a_2)y}{y-1} \right).
\end{align}
This also changes $\lambda_1$ and $\lambda_2$ to $\lambda_1 - \sigma-a_1-a_2+z-\tilde{z}$ and $\lambda_2 - t(\sigma-a_1-a_2+z-\tilde{z})$ respectively.

\begin{cor}
The discrete version of the fifth Painlev\'e equation arises as a discrete isomonodromic deformation, \eqref{deform2}, of \eqref{linear1} with $A(x)$ given by \eqref{ABfactor} and \eqref{Bform4} for 
\[
R(x) = (R_r(x)^{S_1})^{-1} B(-x-\sigma),
\]
where $R_r(x)^{S_1}$ denotes the value of $R_r(x)$ with the action of $S_1$ applied to the entries.
\end{cor}

\begin{proof}
By using \eqref{ddcomp} and \eqref{ABfactor} we have that
\begin{align}
\tilde{A}(x) &= R(x+1)B(-x-\sigma-1)^{-1}B(x) R(x)^{-1}\\
 &= (R_r(x+1)^{S_1})^{-1} B(x) B(-x-\sigma)^{-1} R_r(x)^{S_1},
\end{align}
however, we use the fact that 
\[
B(-x-\sigma)^{-1} = \dfrac{B(x)^{S_1}}{(x+a_1+\sigma)(x+a_2+\sigma)(x+a_3+\sigma)(x+a_4+\sigma)},
\]
and note that $R_r(x)^{S_1}$ is symmetric around $\sigma -1$ (as opposed to $\sigma$). This means the above may be written as
\begin{align*}
\tilde{A}(x) =& \dfrac{(x-a_1)(x-a_2)(x-a_3)(x-a_4)}{(x+a_1+\sigma)(x+a_2+\sigma)(x+a_3+\sigma)(x+a_4+\sigma)} \\
& \left( B(-x-\sigma)^{S_1} R_r(-x-\sigma)^{S_1} \right)^{-1}  \left( B(x)^{S_1} R_r(x)^{S_1} \right),
\end{align*}
where $B(x)^{S_1} R_r(x)^{S_1} = B(x)^{S_3 \circ S_1}$, giving
\[
 \dfrac{(x-a_3)(x-a_4)}{(x+a_3+\sigma)(x+a_4+\sigma)} \left( B(-x-\sigma)^{S_3\circ S_1} \right) B(x)^{S_3\circ S_1}.
\]
We now note that this formulation is only defined up to some left multiplication by $R_l(x)^{S_1\circ S_3}$ (which now satisfies  $R_l(x)^{S_1\circ S_3} = R_l(-\sigma-x)^{S_1\circ S_3}$), which recovers the remaining part of the calculation, showing
\[
\tilde{A}(x) = \left( B(-x-\sigma)^{S_2 \circ S_3 \circ S_1} \right)^{-1} B(x)^{S_2 \circ S_3 \circ S_1}.
\]
Since each transformation, \eqref{S1}, \eqref{S2} and \eqref{S3}, has a simple effect on the parameters $y$ and $z$, it is easy to see that \eqref{dP5} arises as the composition of \eqref{S1} followed by \eqref{S3} and then \eqref{S2}. 
\end{proof}

\section{The list of symmetric Lax pairs}

We simplify the above situation a little by letting $\sigma = 0$ and by letting $\lambda_1 = \lambda_2 = 0$. These simplifications are perfectly valid for considering the difference differential Lax pairs. These values only change when one is calculating the symmetries, as we have done for \eqref{P6} above but we do not present for the cases below. 

Under these simplifying assumptions, each of the Lax pairs we present takes the form \eqref{linear1} in which $A(x)$ is given by
\[
A(x) = B(-x-1)^{-1} B(x),
\]
where $B(x)$ takes the form
\begin{equation}\label{Laxgeneral}
B(x) = \begin{pmatrix} 
x(x-z) + y_1 & x-z \\ 
b_0 + b_1x + b_2 x^2 + b_3 x^3 & x(x-z) + y_2
\end{pmatrix}, 
\end{equation}
where the $b_i$ and $y_j$ for $i = 1,2,3,4$ and $j = 1,2$ are specified up to one parameter by a determinant. 

\subsection{A Lax pair for \eqref{P5}} The corresponding Lax pair for \eqref{P5} is specified by
\[
\det B(x) = t(x-a_1)(x-a_2)(x-a_3).
\]
This gives linear conditions on each of the $b_i$, which may be solved to give
\begin{align*}
b_0 &= -t \left(\left(a_2-z\right) \left(a_3-z\right)+a_1 \left(a_2+a_3-z\right)\right),\\
b_1 &=t(a_2+a_3) + y_1 + y y_2,\\
b_2 &= -t-z,  \hspace{2cm} b_3 = 1.
\end{align*}
Our choice for $y_1$ and $y_2$ are given by
\begin{equation}
y_1 = (1-y) \left(z-a_2\right) \left(z-a_3\right), \hspace{1cm} y_2 = \frac{t \left(z-a_1\right)}{1-y}.
\end{equation}
The isomonodromic deformations are also given by \eqref{comp} where 
\begin{align}
B_{l}(x) =& \begin{pmatrix} 
0 & -\dfrac{1}{t} \\
 -\dfrac{x (x+1)+t \left(a_2+a_3-z\right)+(y+1) y_2}{t} & 1+\dfrac{1}{t} 
 \end{pmatrix},\\
 B_{r}(x) =& \begin{pmatrix} 
  1 & -\dfrac{1}{t} \\
 -\dfrac{x^2+t \left(a_2+a_3-z\right)+2 y_1+(y-1) y_2}{t} & 0
 \end{pmatrix}.
\end{align}
From \eqref{comp} we obtain
\begin{subequations}\label{P5H}
\begin{align}
y' &= \dfrac{(a_2+a_3) (y-1)^2-t y}{t}-\dfrac{2 (y-1) y z}{t},\\
z' &= \dfrac{y_2-y_1}{t}.
\end{align}
\end{subequations}
The system \eqref{P5H} written in terms of $y$ alone give \eqref{P5} for the parameters 
\begin{align}
\alpha = \dfrac{(a_2-a_3)^2}{2}, \hspace{1cm} \beta = -\dfrac{(a_2+a_3)^2}{2}, \hspace{1cm} \gamma = -1-2a_1.
\end{align}

\subsection{A Lax pair for \eqref{P31}} To make a correspondence with \eqref{P31}, we find that the relevant time variable is given by the square root of the leading term, i.e., our associated linear problem is still specified by a matrix, $B(x)$, of the form \eqref{Laxgeneral} where we take the determinant to be
\[
\det B(x) = t^2(x-a_1)(x-a_2).
\]
which specifies linear conditions for $b_0$, \ldots, $b_3$ which are satisfied by
\begin{align*}
b_0 = t^2(a_1+a_2-z), \hspace{1cm} b_1 =y_1+y_2 -t^2,\\
b_2 = -z,  \hspace{2cm} b_3 = 1,\hspace{1cm}
\end{align*}
where $y_1$ and $y_2$ are
\begin{align}
y_1 = ty(z-a_2), \hspace{1cm}y_2= \dfrac{t(z-a_1)}{y}.
\end{align}
The relevant compatibility condition is given by \eqref{comp} where
\begin{align}
B_l(x) &= \begin{pmatrix}  0 & -\dfrac{2}{t} \\
 \dfrac{2 \left(t^2-x (x+1)-2 y_2\right)}{t} & \dfrac{2}{t}
 \end{pmatrix},\\
B_r(x) &= \begin{pmatrix} 
 0 & -\frac{2}{t} \\
- \dfrac{2 x^2-2 t^2+4 y_1}{t} & 0 
 \end{pmatrix}.
\end{align}
Computing the entries of \eqref{comp} reveals the system of first order equations
\begin{subequations}
\begin{align}
y' &= \frac{y (2 t y-4 z-1)+2 t}{t},\\
z' &= \dfrac{2 \left(y_2-y_1\right)}{t},
\end{align}
\end{subequations}
which is equivalent to $y$ satisfying \eqref{P31} for
\begin{align}
\alpha = -8a_2, \hspace{1cm} \beta = 4(1+2a_1), \hspace{1cm} \gamma = 4, \hspace{1cm} \delta= -4.
\end{align}
\subsection{A Lax pair for \eqref{P32}} The associated linear problem is of the form \eqref{Laxgeneral} where
\[
\det B(x) = t(x-a_1).
\]
This condition determines that $b_0$, \ldots, $b_3$ are
\begin{align*}
b_0 = -t, \hspace{1cm} b_1 =\dfrac{ta_1 - y^2-tz}{y} , \hspace{1cm} b_2 = -z,  \hspace{2cm} b_3 = 1,
\end{align*}
with 
\begin{equation}
y_1 = -y, \hspace{1cm} y_2 = -\dfrac{t(z-a_1)}{y}.
\end{equation}
The left and right transformation matrices in \eqref{comp} are
\begin{align}
B_l(x) &= \begin{pmatrix} 
0 & -\dfrac{1}{t} \\
 \dfrac{2 t z-x (x+1) y-2 t a_1}{t y} & \dfrac{1}{t}
 \end{pmatrix},\\
 B_r(x) &= \begin{pmatrix}
   0 & -\dfrac{1}{t} \\
 \dfrac{2 y-x^2}{t} & 0 
 \end{pmatrix}.
 \end{align}
Using these matrices in \eqref{comp} results in the following system of first order differential equations
\begin{subequations}
\begin{align}
y' &= -1-\dfrac{2yz}{t},\\
z' &= \dfrac{y_2 -y_1}{t},
\end{align}
\end{subequations}
which implies that $y$ satisfies \eqref{P32} for the parameter
\[
\beta = -4-8a_1.
\]

\subsection{A Lax pair for \eqref{P33}} The associated linear problem for this last case is of the form \eqref{Laxgeneral} where we make a correspondence with \eqref{P33} when
\[
\det B(x) = \dfrac{t}{4}.
\]
Comparing coefficients imposes enough constraints to determine that $b_0$, \ldots, $b_3$ are 
\begin{align*}
b_0 = 0, \hspace{1cm} b_1 =\dfrac{t+y^2}{2y} , \hspace{1cm} b_2 = -z,  \hspace{2cm} b_3 = 1
\end{align*}
with $y$ being defined by
\[
y_1 = \dfrac{y}{2} , \hspace{1cm} y_2 = \dfrac{t}{2y}.
\]
Under this choice we determine that $R_l(x)$ and $R_r(x)$ are
\begin{align}
B_l(x) &= \begin{pmatrix} 
 0 & -\dfrac{1}{t} \\
 -\dfrac{x (x+1)}{t}-\dfrac{1}{y} & \dfrac{1}{t} \\
 \end{pmatrix},\\
B_r(x) &= \begin{pmatrix}
0 & -\dfrac{1}{t} \\
-\dfrac{x^2}{t}-\dfrac{y}{t} & 0 \\
 \end{pmatrix}.
 \end{align}
These matrices, when used in \eqref{comp}, give the system
\begin{subequations}
\begin{align}
y' &= -\dfrac{2yz}{t},\\
z' &= \dfrac{y_2 -y_1}{t}.
\end{align}
\end{subequations}

\section{A symmetric version of Painlev\'e VI}

The aim of this section is to determine a version of \eqref{P6}. For this, we simply take a slightly different form of $B(x)$ from the previous sections, in which
\begin{equation}
B(x) = \begin{pmatrix} x^2 - z^2 + y & x-z\\ 
b_0 + b_1x + b_2 x^2 + b_3 x^3 & \dfrac{x^2 - z^2}{t} + y_2 
\end{pmatrix}.
\end{equation}
The coefficients of that determinant by
\begin{equation}
\det B(x) = t\sum_{k=0}^{4}  m_k x^k
\end{equation}
so that the coefficients are symmetric in the roots of $\det B(x)$. Following the same steps as above we have
\begin{align*}
b_0 &=\dfrac{m_0 t+\left(y-z^2\right) \left(z^2-y_2\right)}{z}, \\
b_1 &=\dfrac{m_1 t z+m_0 t+\left(y-z^2\right) \left(z^2-y_2\right)}{z^2},\\
b_2 &=z-t \left(m_4 z+m_3\right),\\
b_3 &=1- t m_4,
\end{align*}
where 
\[
y_2 = \dfrac{\det B(z)}{y}.
\]
With this choice, we find that the left and right matrices are
\begin{align}
B_l(x) &= \begin{pmatrix}
\dfrac{z}{t} & -\dfrac{1}{t b_3} \\
\dfrac{t m_2-z(b_2+b_3-2z)-y-y_2}{tb_3} - \dfrac{x(x+1)}{t} & \dfrac{b_3-b_2}{b_3t} \\
\end{pmatrix},\\
B_r(x) &= \begin{pmatrix}
- \dfrac{b_2}{t b_3} & -\dfrac{1}{t b_3} \\
 \dfrac{t m_2 - z(b_2 - 2z)-y - y_2}{tb_3}-\dfrac{x^2}{t} & \dfrac{z}{t} \\
 \end{pmatrix},
\end{align}
whose compatibility results in the following version of \eqref{P6}
\begin{subequations}
\begin{align}
y' &= \dfrac{1}{m_4 t-1}\left( y \left(m_3 + 2\left(\frac{1}{t} + m_4 \right)z \right)- \left. \dfrac{1}{t}\dfrac{\mathrm{d}}{\mathrm{d}x} \det B(x) \right|_{x = z}\right),\\
z' &= \dfrac{1}{m_4 t-1}\left( \dfrac{y}{t} - \dfrac{\det B(z)}{ty} \right).
\end{align}
\end{subequations}
First of all, we recover \eqref{P6} if we make the substitution $Y = (z-a_1)(z-a_2)/y$ where $a_1$ and $a_2$ are roots of $\det B(x)$, then we obtain a  symmetric version of \eqref{P5} when $m_4 = 0$, and a symmetric version of \eqref{P31} when $m_4 = m_3 = 0$ and so on. At each stage there is a different transformation that relates the resulting equations to the versions of \eqref{P5} to \eqref{P33} that appear in the introduction.

\section{Discussion}

One of the interesting features of this work is that the symmetric Lax pairs seem to have a larger group of symmetries than that of the non-symmetric cases. For the non-symmetric cases, the relevant group acting on the parameter space has the structure of a lattice, whereas in the symmetric case, we have a group generated by a pairs of transformations, each pair is the generator of some infinite dihedral group. This setting is much closer to the affine Weyl symmetries of the Painlev\'e equations \cite{MR1957556} and the discrete Painlev\'e equations \cite{Noumi:AffinedPs}.

\section{Acknowledgements}

We would like to thank Frank Nijhoff for his correspondence and bringing our attention to \cite{Adler1994}.
\bibliographystyle{plain}%

\bibliography{../../refs}

\end{document}